\documentclass[a4paper]{article}
\pdfoutput=1

\usepackage[left=2.5cm,right=2.5cm,top=2.5cm,bottom=2.5cm]{geometry}
\usepackage{bbm}
\usepackage{graphicx}

\usepackage{amssymb}
\usepackage{amsthm}
\usepackage{amsmath}

\usepackage{xspace}

\theoremstyle{plain}

\newtheorem{theorem}{Theorem}[section]

\newtheorem{claim}[theorem]{Claim}

\newtheorem{conjecture}[theorem]{Conjecture}

\usepackage{subcaption}
\usepackage{tikz}
\usetikzlibrary{calc}
\usetikzlibrary{decorations.markings}
\usetikzlibrary{shapes.geometric,positioning}
\usetikzlibrary{arrows}
\tikzset{vertex/.style={minimum size=2mm,circle,fill=black,draw,inner sep=0pt},
         decoration={markings,mark=at position .5 with {\arrow[black,thick]{stealth}}}}

\theoremstyle{definition}
\newtheorem{definition}[theorem]{Definition}

\usepackage[boxed,commentsnumbered]{algorithm2e}
\newcommand{\alg}{\ensuremath{\textsc{Alg}_\beta}\xspace}

\usepackage{subfig}
\newtheorem{myrule}{Rule}

\newcommand{\TSP}{\ensuremath{\beta\mbox{-TSP}}\xspace}

\newcommand{\Sodd}{\ensuremath{S^{\mathrm{odd}}}}
\newcommand{\Seven}{\ensuremath{S^{\mathrm{even}}}}

\newcommand{\bi}[2]{\ensuremath{\{#1^-,#2^-\}}\xspace}

\newcommand{\indeg}[1]{\ensuremath{\mathrm{deg}^-(#1)}\xspace}
\newcommand{\outdeg}[1]{\ensuremath{\mathrm{deg}^+(#1)}\xspace}
\newcommand{\MG}{\ensuremath{D}\xspace}
\newcommand{\ie}{i.\,e.\xspace}

\title{An Improved Approximation Algorithm for the Traveling
    Salesman Problem with\\ Relaxed Triangle Inequality\thanks{This work was partially supported by ERC advanced investigator grant
    226203 and Deutsche Forschungsgemeinschaft grant BL511/10-1}}

\author{Tobias M\"{o}mke\\
        Department of Computer Science\\ 
        Saarland University\\ 
        {\tt moemke@cs.uni-saarland.de}
        }
\begin{document}
\maketitle

\begin{abstract}
    Given a complete edge-weighted graph $G$, we present a polynomial time algorithm 
    to compute a degree-four-bounded spanning Eulerian subgraph of $2G$ that has at
    most $1.5$ times the weight of an optimal TSP solution of $G$.
    Based on this algorithm and a novel use of orientations in graphs, we obtain
    a $(3\beta/4 + 3\beta^2/4)$-approximation algorithm for TSP with $\beta$-relaxed triangle inequality ($\beta$-TSP), where $\beta \ge 1$.
    A graph $G$ is an instance of $\beta$-TSP, if it is a complete graph with edge weights $c\colon E(G) \rightarrow \mathbb{Q}_{\ge 0}$ that are restricted as follows.
    For each triple of vertices $u,v,w \in V(G)$, $c(\{u,v\}) \le \beta (c(\{u,w\}) + c(\{w,v\}))$.
\end{abstract}

\section{Introduction}\label{sec:intro}

In the traveling salesman problem we are given a complete edge-weighted graph and we
have to find a minimum-weight Hamiltonian tour, \ie, a tour that visits each vertex exactly once. 
This classical problem has been studied extensively
and in many variations. Most of the variations concern restrictions of the weight
function. One of the most natural restrictions is to assume that the weight
function is a metric. Intuitively this means that we are allowed to
take the shortest path to the next vertex to visit, even if this means to visit
some of the vertices more than once. 

Despite intensive research for more than 30 years, Christofides' algorithm
is still the best known approximation algorithm for the metric
traveling salesman problem and its approximation ratio is 1.5~\cite{chr76}. For similar
settings, however, a recent fast development has started. For graphic
metrics---metrics obtained by taking the lengths of the shortest paths in an
unweighted graph as weights---a sequence of improvements was published within a
short time frame. 

The development started with papers concerning graphs with degree restrictions
and connectivity requirements \cite{GLS05,BSvdSS12}. For the general graph-TSP,
the first improvements over Christofides' algorithm were due to Oveis Gharan,
Saberi and Singh \cite{GSS11a} and M\"{o}mke and Svensson \cite{MS11a}. The
analysis of M\"{o}mke and Svensson was subsequently refined by
Mucha \cite{Muc12}. The currently best approximation algorithm achieves an
approximation ratio of $1.4$ and is due to Seb\H{o} and Vygen \cite{SV14},
combining results from \cite{MS11a} with a smart technique based on ear
decompositions.

Instead of restricting the metric, in this paper we consider a
relaxation. We use a relaxation parameter $\beta \ge 1$ and require that the
input instance is a complete graph $G$ where the non-negative weight function 
$c\colon E(G) \rightarrow \mathbb{Q}_{\ge 0}$ satisfies the relaxed triangle 
inequality $c(\{u,w\}) \le \beta (c(\{u,v\}) + c(\{v,w\}))$ for any three
vertices $u,v,w \in V(G)$.

For approximation algorithms, the relaxed triangle inequality was introduced by 
Bandelt, Crama, and Spieksma \cite{BCS94}. This type of parameterization also
provides a suitable type of relaxation in the context of stability of
approximation~\cite{BHKSU02} and our result fits into this framework.

The $\beta$-relaxed version of metric TSP (\TSP) was first considered by 
Andreae and Bandelt \cite{AB95} who presented a $1.5 \beta^2 + 0.5 \beta$ approximation algorithm.
Subsequently Andreae improved the result to $\beta + \beta^2$ \cite{And01}. The next
development was due to B\"{o}ckenhauer et al.~\cite{BHKSU02}. They obtained a $1.5
\beta^2$-approximation algorithm, which is better than the previous algorithms
for $1 < \beta <  2$. Bender and Chekuri \cite{BC00} independently obtained a $4\beta$ approximation algorithm, 
which is better than the algorithm of Anreae for $\beta > 3$.

\subsection{Results and Overview of Techniques}

We provide an improved approximation algorithm for \TSP.
\begin{theorem}\label{thm:main}
    There is a polynomial time $(3\beta/4 + 3\beta^2/4)$-approximation algorithm
    for \TSP.
\end{theorem}
The approximation ratio of our algorithm outperforms the ratios provided by Andreae \cite{And01} and 
by B\"{o}ckenhauer et al.~\cite{BHKSU02} for all values of $\beta$. 
The $4\beta$ approximation algorithm of Bender and Chekuri \cite{BC00} still is better
than our algorithm for $\beta > 13/3 \approx 4.33$. To obtain our result we first use the matroid version by Kir\'aly et 
al.~\cite{KLS12_degree} of a bounded degree spanning tree result of Singh and Lau \cite{SL07}, 
combined with special $b$-matchings that respect parities of vertex degrees. 
We obtain a degree-$4$-bounded spanning Eulerian subgraph in $2G$ for any complete graph $G$ with edge weights 
$c\colon E(G) \rightarrow \mathbb{Q}$ 
such that the weight of the computed graph is at most 1.5 times the weight of an optimal TSP solution.
Finally, we introduce an orientation technique that provides a cactus graph with useful properties 
such that the weight for shortcuts within the graph is restricted when constructing the TSP solution.
One key insight is that we obtain two disjoint sets of edges such that we have to consider the factor $\beta^2$ only for the smaller of the two sets.

\section{Preliminaries}\label{sec:prelim}

All graphs in this paper are allowed to have multiple edges.
For convenience of notation, however, we do not distinguish between multi-sets and sets.
We handle multiple edges of a graph as separate edges. 
This way there may be edge-disjoint cycles of length two.
(We define a cycle to be a simple cycle, \ie, vertices may not be visited twice.)

Given a graph $G$, $V(G)$ and $E(G)$ are its set of vertices and its set of edges. 
For a set of edges $F \subseteq E(G)$, we write $c(F)$ as shorthand for $\sum_{e \in F} c(e)$.
Similarly, for a graph $G$ we write $c(G)$ as shorthand for $c(E(G))$.
A \emph{block} of a graph $G$ is a maximal two-vertex-connected subgraph.

A $b$-matching of a graph $G$ is a subgraph $G'$ of $G$ with possible additional multiplicities of edges
where each vertex has a degree of at most $b$. 
Equivalently, we identify a $b$-matching with its characteristic vector $\mathbf{x}$ of edges, 
that is, for each edge $e \in E(G)$, $\mathbf{x}$ has an entry $x_e \in \mathbb{Z}$ and $x_e \ge 1$ if and only if $e \in E(G')$.

Let $G=(V,E)$ be an undirected graph and let $M$ be the $V \times E$ incidence matrix of $G$. 
Let $\mathbf{l} \le \mathbf{m}$ and $\mathbf{a} \le \mathbf{b}$ be integer vectors in $\mathbb{Z}^E$ resp.\ $\mathbb{Z}^V$ 
and let $\Sodd$ and $\Seven$ be disjoint subsets of $V$.

Consider the following constraints that impose restrictions on the characteristic vector of the $b$-matching $\mathbf{x} \in \mathbb{Z}^E$.

\begin{equation}\label{eq:bmatching}
    \begin{array}{cllccll}
        \mathrm{(i)}   & \mathbf{l} \le \mathbf{x} \le \mathbf{m}           &
        \quad&
        \mathrm{(iii)} & (M\mathbf{x})_v \mbox{ is odd}  &\quad \mbox{if } v \in
        \Sodd\\
        \mathrm{(ii)}  & \mathbf{a} \le M\mathbf{x} \le \mathbf{b}          &
        \qquad&
        \mathrm{(iv)}  & (M\mathbf{x})_v \mbox{ is even} &\quad \mbox{if } v \in
        \Seven
    \end{array}
\end{equation}
These constraints specify 
(i) bounds on the multiplicity of edges in $G$, 
(ii) bounds on the degrees of the vertices in $G$, and 
(iii,iv) the parities of degrees for specific vertices (since $M\mathbf{x}$ is the vector of vertex degrees).

We slightly abuse notation and, whenever the meaning is clear from the context,
we associate an integer $\mathbf{i}$ with the corresponding vector $(i,i,\dotsc,i)$.
Note that for an integer $b'$, \eqref{eq:bmatching} specifies a $b'$-matching if we set $\mathbf{l} = \mathbf{0}$,
$\mathbf{m} = \boldsymbol{\infty}$, $\mathbf{a}=\mathbf{0}$, $\mathbf{b} = \mathbf{b'}$, 
and $\Sodd=\Seven=\emptyset$. 
In the context of this paper we need specific $b$-matchings that we will specify by giving values to the parameters of \eqref{eq:bmatching}. 

The following theorem is Theorem~$36.5$ in Schrijver's book \cite{Sch03}.
(Here, we use a simplified setting. In the original theorem a more general class of graphs can be used.)
\begin{theorem}[Edmonds, Johnson \cite{EJ73}]\label{thm:bmatching}
    For any $\mathbf{c} \in \mathbbm{Q}^E$, an integer vector $\mathbf{x}$
    maximizing $\mathbf{c}^T\mathbf{x}$ over
    \eqref{eq:bmatching} can be found in strongly polynomial time (if it exists).
\end{theorem}

Given a graph $G$, a $1$-tree is a subgraph of $G$ composed of a spanning tree on $V(G) \setminus \{v_1\}$ for some $v_1 \in V(G)$ and two edges incident to $v_1$.
If $G$ is a complete graph, then for each choice of $v_1$ there is a matroid $M$ such that the $1$-trees of $G$ are the bases of $M$~\cite{HK70}.
We are interested in $1$-trees where additionally the vertices have degree restrictions.

For a vector $\mathbf{b}$ of vertex degrees, a $1$-tree $T$ of $G$ is degree $\mathbf{b}$ bounded 
if the degree of each vertex $v \in V(T)$ is at most $\mathbf{b}_v$.
The following theorem follows directly from Kir\'aly et al.~\cite{KLS12_degree} who showed a more general result for matroids.
(We run their algorithm for each choice of $v_1$.)

\begin{theorem}\label{thm:tree}
Given a complete graph $G$ with edge weights $c\colon E(G) \rightarrow \mathbb{Q}$ and a vector $\mathbf{b}$ of upper bounds on the vertex degrees, 
there is polynomial time algorithm that computes a $1$-tree $T$ in $G$ such that 
(a) $T$ is degree $(\mathbf{b}+\mathbf{1})$ bounded and
(b) $c(T)\le c(T')$ for all degree $\mathbf{b}$ bounded $1$-trees $T'$ in $G$.
\end{theorem}

For our algorithm we need cactus graphs. 
Here we use the following definition of cacti which strictly speaking specifies the subclass of 2-edge-connected cacti.
\begin{definition}[Cactus]\label{def:treecycle}
    A graph $G$ is a \emph{cactus} if it is 2-edge connected and 
    no two cycles share an edge.
\end{definition}
In particular, the blocks of a cactus are cycles.

Given a graph $G$, $2G$ is the graph where each edge of $G$ is doubled.
A graph $G'$ is a \emph{spanning Eulerian subgraph of $2G$} if
$V(G) = V(G')$, $G'$ is a connected subgraph of $2G$, and each vertex of $G'$
has even degree.

\section{The Algorithm $\boldmath{\textsc{Alg}_\beta}$}\label{sec:tsp}

Before we proceed to the main result, we show an
intermediate observation. With the preparation in the preliminaries, the following theorem 
is not hard. It is, however, of independent interest given the importance of Eulerian
graphs. Note that the weight function is neither required to be metric nor to be non-negative.

\begin{theorem}\label{thm:eul}
    Let $G$ be a complete graph with edge weights $c\colon E(G) \rightarrow \mathbb{Q}$. 
    There is a polynomial time algorithm that computes a degree-$4$-bounded spanning Eulerian subgraph of $2G$ such that 
    its total weight is at most 1.5 times the weight of an optimal TSP solution in $G$.
\end{theorem}
\begin{proof}
    We combine Theorem~\ref{thm:bmatching} with Theorem~\ref{thm:tree}. 
    To obtain a degree-$4$-bounded Eulerian subgraph we compute a
    degree-$3$-bounded $1$-tree $\tau$ and add a $b$-matching with
    the following parameters of \eqref{eq:bmatching}.
    We set $\mathbf{l}=\mathbf{0}$, $\mathbf{m}=\mathbf{1}$, $\mathbf{a}=\mathbf{0}$, and $\mathbf{b}=\mathbf{2}$.
    Then we set $\Sodd$ to be the set of all vertices of odd degree in $\tau$, and $\Seven$ the set of all remaining vertices.

    By Theorem~\ref{thm:bmatching} we obtain an optimal $b$-matching $M$ with the specified parameters in strongly polynomial time.
    This way, all odd degree vertices in $\tau$ obtain a degree increased by exactly one and all remaining vertices
    (of degree two) either stay unchanged or become degree-$4$-vertices.
 
    Due to Theorem~\ref{thm:tree} we can, in polynomial time, obtain a degree-$3$-bounded
    $1$-tree of at most the weight of an optimal TSP solution,
    because any TSP solution is a degree-2-bounded $1$-tree. 
    For the remaining argument we use that, similar to the matchings in Christofides' algorithm, 
    any TSP solution is composed of two edge-disjoint $b$-matchings obeying the specified restrictions. 
    A minimum weight solution thus has at most half the weight of an optimal TSP solution. 
\end{proof}

In order to state the algorithm, we need a few more definitions.
We use bi-directed arcs, where we mark the tails with $+$ and the heads with $-$.
Thus for an arc with two heads $u$ and $v$, we write \bi{u}{v}.
For an arc with tail $u$ and head $v$, we write $(u,v)$ instead of
$\{u^+,v^-\}$. 
In the context of this paper, there are no arcs $\{u^+,v^+\}$.
Sometimes, we are interested in the vertices of arcs, independent of directions.
In these cases, we refer to arcs as edges.

Let $A$ be the arc set and $V$ the vertex set of a bi-directed graph \MG.
Let $v \in V$ be a vertex of \MG. 
The in-degree of $v$ is $\indeg{v} = |\{a \in A : \exists u \in V, a = (u,v) \mbox{ or } a = \bi{u}{v}\}|$.
Similarly, $\outdeg{v} = |\{a \in A : \exists u \in V, a = (v,u) \}|$.

We design an algorithm \alg that relies on graph transformations specified by the following rule.
\begin{myrule}\label{rule:join}
    Let \MG be a bi-directed graph and let $v$ be a vertex given as input such 
    that $\outdeg{v}=2$. Let $w \neq w'$ be the vertices such that the arcs $(v,w),(v,w')$
    are in \MG. We remove both $(v,w)$ and $(v,w')$ from \MG and add the arc
    $\bi{w}{w'}$.
\end{myrule}

We also need the reverse operation of Rule~\ref{rule:join}. If we
obtained $\bi{w}{w'}$ due to an operation of Rule~\ref{rule:join} to $v$, to reverse
Rule~\ref{rule:join} at $\bi{w}{w'}$ 
means that we remove $\bi{w}{w'}$ and add the two arcs $(v,w),(v,w')$.

Let $G'$ be a bi-directed graph. 
Suppose we are given a set of vertices $S \subset V(G')$ and vertices $w,w' \in S$, $v \in V(G') \setminus S$ with the following properties:
(i) there are no arcs of $G'$ between $S$ and $V(G') \setminus S$; 
(ii) each connected component of $G'$ is a cactus;
(iii) $w$ and $w'$ are connected by an arc $\bi{w}{w'}$ in $G'$, obtained due to an application of Rule~\ref{rule:join} to $v$.
Then we say that $\bi{w}{w'}$ is an \emph{entry site} of the subgraph of $G'$ induced by $S$.
Let $C$ be the block of $G'$ that contains $v$. 
Then after reversing the application of Rule~\ref{rule:join} at $\bi{w}{w'}$, $v$ is an \emph{exit point of $C$} (but not of any other block). 

Fig.~\ref{fig:algorithm}(c),(d) shows an example of an entry site and an exit point.
The reason to consider entry sites and exit points is to create a cactus such
that its blocks form a specific hierarchy that allows us to exclude long (and
therefore expensive) shortcuts when constructing a Hamiltonian tour from the
cactus.

\begin{algorithm}[bt]
\TitleOfAlgo{\alg}
  \label{alg:tsp}
  \SetAlgoNoLine
  \SetKwInOut{Input}{Input}\SetKwInOut{Output}{Output}
  \Input{An instance $G$ of \TSP.}
  \Output{A Hamiltonian tour in $G$.}
  Apply Theorem~\ref{thm:eul} to obtain a degree-$4$-bounded spanning Eulerian
  subgraph $H$ of $2G$\;
  Find an Eulerian tour in $H$ and orient the edges along the tour\;
  Apply Rule~\ref{rule:join} to all vertices of degree $4$ in $H$. Let $H'$ be the
  resulting graph \tcp*{The graph $H'$ is collection of disjoint cycles.}
  Choose a cycle $\hat{C}$ in $H'$ and set $K = \hat{C}$\tcp*{$K$ is a cactus.}
  \While{$V(K) \neq V(G)$}{
      Determine an entry site $a$ of $K$\;
      Reverse the application of Rule~\ref{rule:join} at $a$\;
      Update $K$ to include the additional cycle as a new block\;
  }
  Let $\hat{C}'$ be the block in $K$ that contains $V(\hat{C})$\;
  Let $K'$ be $K$ with all orientations of edges removed\; 
  \For{each block $C$ of $K'$ except $\hat{C}'$}{
      Let $v$ be the exit point of $C$ (w.\,r.\,t.~$K$)\;
      Let $e,e'$ be the two edges incident to $v$ in $C$ such that $c(e) \le
      c(e')$\;
      Orient the edges along $C$ in $K'$, starting with $e$ from $v$\;
  }
  Orient the edges along $\hat{C}'$ in $K'$ (in arbitrary direction)\;
  Apply Rule~\ref{rule:join} to all vertices of degree $4$ in $K'$\; 
  Return the resulting tour $K''$ (without orientations).
\end{algorithm}
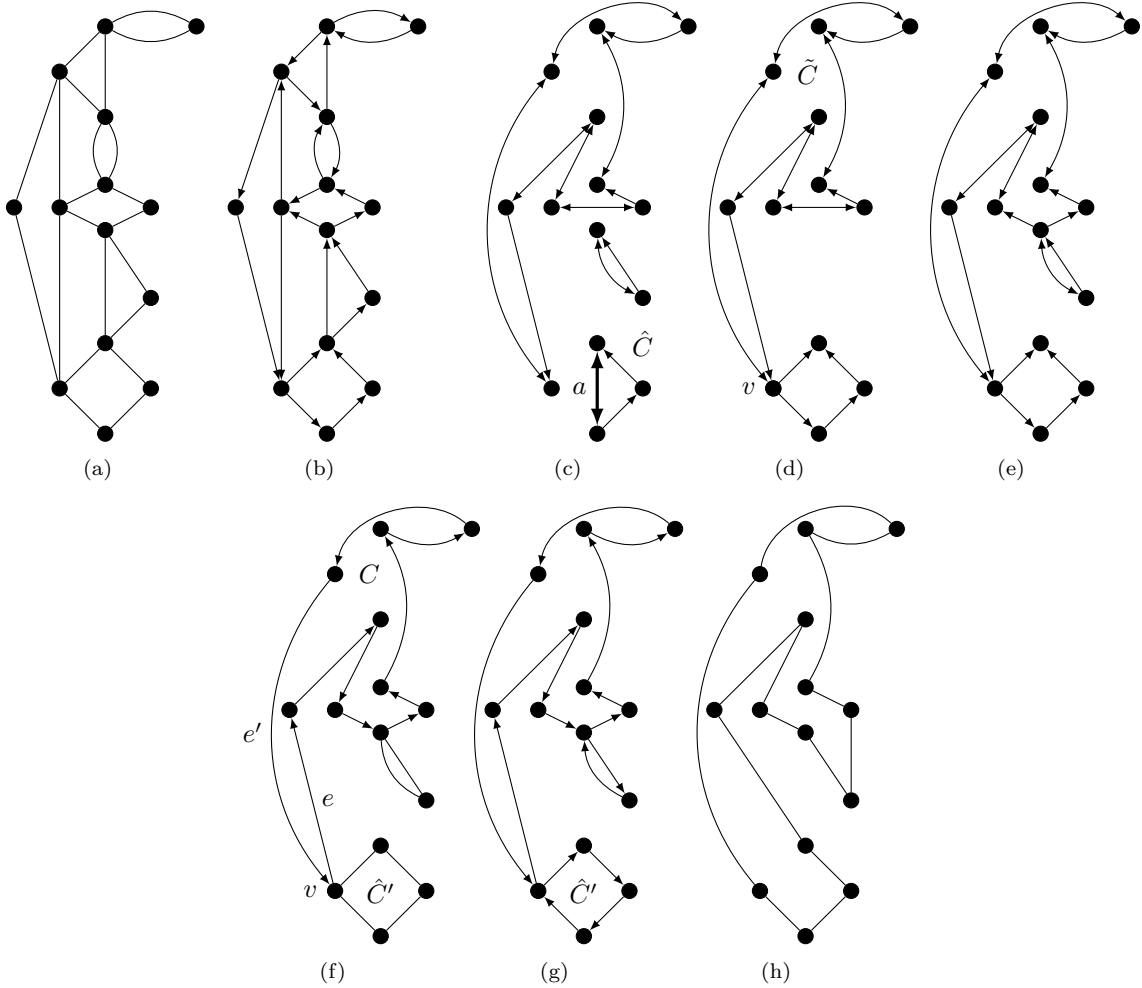
\begin{figure}[bt]
  \centering
  \begin{subfigure}[b]{2.8cm}
    \centering
    \begin{tikzpicture}[scale=0.6]
      \node (1) at (2,0) [vertex]{};
      \node (2) at (3,1) [vertex]{};
      \node (3) at (1,1) [vertex]{};
      \node (4) at (2,2) [vertex]{};
      \node (5) at (3,3) [vertex]{};
      \node (6) at (2,4.5) [vertex]{};
      \node (7) at (3,5) [vertex]{};
      \node (8) at (1,5) [vertex]{};
      \node (9) at (0,5) [vertex]{};
      \node (10) at (2,5.5) [vertex]{};
      \node (11) at (2,7) [vertex]{};
      \node (12) at (1,8) [vertex]{};
      \node (13) at (4,9) [vertex]{};
      \node (14) at (2,9) [vertex]{};
      \draw[](3)--(4);
      \draw[](4)--(6);
      \draw[](6)--(8);
      \draw[](8)--(3);
      \draw[](3)--(1);
      \draw[](1)--(2);
      \draw[](2)--(4);
      \draw[](4)--(5);
      \draw[](5)--(6);
      \draw[](6)--(7);
      \draw[](7)--(10);
      \draw[](10) to[bend left] (11);
      \draw[](11)--(14);
      \draw[](14) to[bend left] (13);
      \draw[](13) to[bend left] (14);
      \draw[](14)--(12);
      \draw[](12)--(11);
      \draw[](11) to[bend left] (10);
      \draw[](10)--(8);
      \draw[](8)--(12);
      \draw[](12)--(9);
      \draw[](9)--(3);
    \end{tikzpicture}
    \caption{$~$}
    \label{fig:H}
  \end{subfigure}
  \begin{subfigure}[b]{2.8cm}
    \centering
    \begin{tikzpicture}[->,scale=0.6]
      \node (1) at (2,0) [vertex]{};
      \node (2) at (3,1) [vertex]{};
      \node (3) at (1,1) [vertex]{};
      \node (4) at (2,2) [vertex]{};
      \node (5) at (3,3) [vertex]{};
      \node (6) at (2,4.5) [vertex]{};
      \node (7) at (3,5) [vertex]{};
      \node (8) at (1,5) [vertex]{};
      \node (9) at (0,5) [vertex]{};
      \node (10) at (2,5.5) [vertex]{};
      \node (11) at (2,7) [vertex]{};
      \node (12) at (1,8) [vertex]{};
      \node (13) at (4,9) [vertex]{};
      \node (14) at (2,9) [vertex]{};
      \draw[>=latex](3)--(4);
      \draw[>=latex](4)--(6);
      \draw[>=latex](6)--(8);
      \draw[>=latex](8)--(3);
      \draw[>=latex](3)--(1);
      \draw[>=latex](1)--(2);
      \draw[>=latex](2)--(4);
      \draw[>=latex](4)--(5);
      \draw[>=latex](5)--(6);
      \draw[>=latex](6)--(7);
      \draw[>=latex](7)--(10);
      \draw[>=latex](10) to[bend left] (11);
      \draw[>=latex](11)--(14);
      \draw[>=latex](14) to[bend left] (13);
      \draw[>=latex](13) to[bend left] (14);
      \draw[>=latex](14)--(12);
      \draw[>=latex](12)--(11);
      \draw[>=latex](11) to[bend left] (10);
      \draw[>=latex](10)--(8);
      \draw[>=latex](8)--(12);
      \draw[>=latex](12)--(9);
      \draw[>=latex](9)--(3);
    \end{tikzpicture}
    \caption{$~$}
    \label{fig:Eulerian}
  \end{subfigure}
  \quad
  \begin{subfigure}[b]{2.8cm}
    \centering
    \begin{tikzpicture}[->,scale=0.6]
      \node (1) at (2,0) [vertex]{};
      \node (2) at (3,1) [vertex]{};
      \node (3) at (1,1) [vertex]{};
      \node (4) at (2,2) [vertex]{};
      \node (5) at (3,3) [vertex]{};
      \node (6) at (2,4.5) [vertex]{};
      \node (7) at (3,5) [vertex]{};
      \node (8) at (1,5) [vertex]{};
      \node (9) at (0,5) [vertex]{};
      \node (10) at (2,5.5) [vertex]{};
      \node (11) at (2,7) [vertex]{};
      \node (12) at (1,8) [vertex]{};
      \node (13) at (4,9) [vertex]{};
      \node (14) at (2,9) [vertex]{};
      \node at (3,2) {$\hat{C}$};
      \draw[>=latex,<->](5) to[bend left] (6);
      \draw[>=latex, <->](3) to[bend left=40] (12);
      \draw[>=latex,<->,very thick](1) -- node[left]{$a$} (4);
      \draw[>=latex](1)--(2);
      \draw[>=latex](2)--(4);
      \draw[>=latex](5)--(6);
      \draw[>=latex,<->](7)--(8);
      \draw[>=latex](7)--(10);
      \draw[>=latex,<->](8) -- (11);
      \draw[>=latex,<->](10) to[bend right] (14);
      \draw[>=latex,<->](12) to[bend left=60] (13);
      \draw[>=latex](13) to[bend left] (14);
      \draw[>=latex,<->](9)--(11);
      \draw[>=latex](9)--(3);
    \end{tikzpicture}
    \caption{$~$}
    \label{fig:before-while}
  \end{subfigure}
  \begin{subfigure}[b]{2.8cm}
    \centering
    \begin{tikzpicture}[->,scale=0.6]
      \node (1) at (2,0) [vertex]{};
      \node (2) at (3,1) [vertex]{};
      \node (3) at (1,1) [vertex, label = left:$v$]{};
      \node (3) at (1,1) [vertex]{};
      \node (4) at (2,2) [vertex]{};
      \node (7) at (3,5) [vertex]{};
      \node (8) at (1,5) [vertex]{};
      \node (9) at (0,5) [vertex]{};
      \node (10) at (2,5.5) [vertex]{};
      \node (11) at (2,7) [vertex]{};
      \node (12) at (1,8) [vertex]{};
      \node (13) at (4,9) [vertex]{};
      \node (14) at (2,9) [vertex]{};
      \node at (1.75,8) {$\tilde{C}$};
      \draw[>=latex, <->](3) to[bend left=40] (12);
      \draw[>=latex](1)--(2);
      \draw[>=latex](2)--(4);
      \draw[>=latex,<->](7)--(8);
      \draw[>=latex](7)--(10);
      \draw[>=latex,<->](8) -- (11);
      \draw[>=latex,<->](10) to[bend right] (14);
      \draw[>=latex,<->](12) to[bend left=60] (13);
      \draw[>=latex](13) to[bend left] (14);
      \draw[>=latex,<->](9)--(11);
      \draw[>=latex](9)--(3);
      \draw[>=latex](3)--(1);
      \draw[>=latex](3)--(4);
    \end{tikzpicture}
    \caption{$~$}
    \label{fig:while}
  \end{subfigure}
  \begin{subfigure}[b]{2.8cm}
    \centering
    \begin{tikzpicture}[->,scale=0.6]
      \node (1) at (2,0) [vertex]{};
      \node (2) at (3,1) [vertex]{};
      \node (3) at (1,1) [vertex]{};
      \node (4) at (2,2) [vertex]{};
      \node (5) at (3,3) [vertex]{};
      \node (6) at (2,4.5) [vertex]{};
      \node (7) at (3,5) [vertex]{};
      \node (8) at (1,5) [vertex]{};
      \node (9) at (0,5) [vertex]{};
      \node (10) at (2,5.5) [vertex]{};
      \node (11) at (2,7) [vertex]{};
      \node (12) at (1,8) [vertex]{};
      \node (13) at (4,9) [vertex]{};
      \node (14) at (2,9) [vertex]{};
      \draw[>=latex,<->](5) to[bend left] (6);
      \draw[>=latex, <->](3) to[bend left=40] (12);
      \draw[>=latex](1)--(2);
      \draw[>=latex](2)--(4);
      \draw[>=latex](5)--(6);
      \draw[>=latex](7)--(10);
      \draw[>=latex,<->](8) -- (11);
      \draw[>=latex,<->](10) to[bend right] (14);
      \draw[>=latex,<->](12) to[bend left=60] (13);
      \draw[>=latex](13) to[bend left] (14);
      \draw[>=latex,<->](9)--(11);
      \draw[>=latex](9)--(3);
      \draw[>=latex](3)--(1);
      \draw[>=latex](3)--(4);
      \draw[>=latex](6)--(7);
      \draw[>=latex](6)--(8);
    \end{tikzpicture}
    \caption{$~$}
    \label{fig:after-while}
  \end{subfigure}
  \begin{subfigure}[b]{2.8cm}
    \centering
    \begin{tikzpicture}[scale=0.6]
      \node (1) at (2,0) [vertex]{};
      \node (2) at (3,1) [vertex]{};
      \node (3) at (1,1) [vertex, label=left:$v$]{};
      \node (4) at (2,2) [vertex]{};
      \node (5) at (3,3) [vertex]{};
      \node (6) at (2,4.5) [vertex]{};
      \node (7) at (3,5) [vertex]{};
      \node (8) at (1,5) [vertex]{};
      \node (9) at (0,5) [vertex]{};
      \node (10) at (2,5.5) [vertex]{};
      \node (11) at (2,7) [vertex]{};
      \node (12) at (1,8) [vertex]{};
      \node (13) at (4,9) [vertex]{};
      \node (14) at (2,9) [vertex]{};
      \node at (2,1) {$\hat{C}'$};
      \node at (1.75,8) {$C$};
      \draw[>=latex](5) to[bend left] (6);
      \draw[>=latex,->](12) to[bend right=40] node[left]{$e'$} (3);
      \draw[>=latex](2)--(1);
      \draw[>=latex](4)--(2);
      \draw[>=latex](5)--(6);
      \draw[>=latex,->](7)--(10);
      \draw[>=latex,->](11) -- (8);
      \draw[>=latex,->](10) to[bend right] (14);
      \draw[>=latex,->](13) to[bend right=60] (12);
      \draw[>=latex,->](14) to[bend right] (13);
      \draw[>=latex,->](9)--(11);
      \draw[>=latex,->](3)-- node[right]{$e$} (9);
      \draw[>=latex](1)--(3);
      \draw[>=latex](3)--(4);
      \draw[>=latex,->](6)--(7);
      \draw[>=latex,->](8)--(6);
    \end{tikzpicture}
    \caption{$~$}
    \label{fig:for}
  \end{subfigure}
  \begin{subfigure}[b]{2.8cm}
    \centering
    \begin{tikzpicture}[scale=0.6]
      \node (1) at (2,0) [vertex]{};
      \node (2) at (3,1) [vertex]{};
      \node (3) at (1,1) [vertex]{};
      \node (4) at (2,2) [vertex]{};
      \node (5) at (3,3) [vertex]{};
      \node (6) at (2,4.5) [vertex]{};
      \node (7) at (3,5) [vertex]{};
      \node (8) at (1,5) [vertex]{};
      \node (9) at (0,5) [vertex]{};
      \node (10) at (2,5.5) [vertex]{};
      \node (11) at (2,7) [vertex]{};
      \node (12) at (1,8) [vertex]{};
      \node (13) at (4,9) [vertex]{};
      \node (14) at (2,9) [vertex]{};
      \node at (2,1) {$\hat{C}'$};
      \draw[>=latex,->](5) to[bend left] (6);
      \draw[>=latex,->](12) to[bend right=40] (3);
      \draw[>=latex,->](2)--(1);
      \draw[>=latex,->](4)--(2);
      \draw[>=latex,->](6)--(5);
      \draw[>=latex,->](7)--(10);
      \draw[>=latex,->](11) -- (8);
      \draw[>=latex,->](10) to[bend right] (14);
      \draw[>=latex,->](13) to[bend right=60] (12);
      \draw[>=latex,->](14) to[bend right] (13);
      \draw[>=latex,->](9)--(11);
      \draw[>=latex,->](3)--(9);
      \draw[>=latex,->](1)--(3);
      \draw[>=latex,->](3)--(4);
      \draw[>=latex,->](6)--(7);
      \draw[>=latex,->](8)--(6);
    \end{tikzpicture}
    \caption{$~$}
    \label{fig:after-for}
  \end{subfigure}
  \begin{subfigure}[b]{2.8cm}
    \centering
    \begin{tikzpicture}[scale=0.6]
      \node (1) at (2,0) [vertex]{};
      \node (2) at (3,1) [vertex]{};
      \node (3) at (1,1) [vertex]{};
      \node (4) at (2,2) [vertex]{};
      \node (5) at (3,3) [vertex]{};
      \node (6) at (2,4.5) [vertex]{};
      \node (7) at (3,5) [vertex]{};
      \node (8) at (1,5) [vertex]{};
      \node (9) at (0,5) [vertex]{};
      \node (10) at (2,5.5) [vertex]{};
      \node (11) at (2,7) [vertex]{};
      \node (12) at (1,8) [vertex]{};
      \node (13) at (4,9) [vertex]{};
      \node (14) at (2,9) [vertex]{};
      \draw[](5) -- (7);
      \draw[](3) to[bend left=40] (12);
      \draw[](1)--(2);
      \draw[](2)--(4);
      \draw[](5)--(6);
      \draw[](7)--(10);
      \draw[](8) -- (11);
      \draw[](10) to[bend right] (14);
      \draw[](12) to[bend left=63] (13);
      \draw[](13) to[bend left] (14);
      \draw[](9)--(11);
      \draw[](3)--(1);
      \draw[](4)--(9);
      \draw[](6)--(8);
    \end{tikzpicture}
    \caption{$~$}
    \label{fig:tour}
  \end{subfigure}

  \caption{An example of intermediate graphs in \alg.
  (a) The graph $H$.
  (b) The Eulerian tour in $H$.
  (c) The graph $H'$ before the while loop with entry site $a$ of $\hat{C}$. 
  (d) The graph $K$ after one iteration of the while loop. The vertex $v$ is the exit point of $\tilde{C}$.
  (e) The graph $K$ after the while loop.
  (f) The graph $K'$ after one iteration of the for loop where $v$ is the exit point of $C$.
  (g) The graph $K'$ after the for loop.
  (h) The graph $K''$.}
  \label{fig:algorithm}
\end{figure}

\subsection{Proof of Theorem~\ref{thm:main}}
    We show that \alg
    is a polynomial time 
    $(3\beta/4 + 3 \beta^2/4)$-approx\-imation algorithm for \TSP.

    Let $G$ be an instance of \TSP. We have to verify that \alg 
    indeed computes a Hamiltonian tour (see Fig.~\ref{fig:algorithm}).
    Since $G$ is a complete graph, we can apply Theorem~\ref{thm:eul} to obtain a degree-$4$-bounded spanning Eulerian subgraph $H$ of $2G$. 
    Since $H$ is Eulerian, we can
    efficiently find an Eulerian tour in $H$ and thus orient the edges
    accordingly. We use the following well known property of $H$, observed by Euler.
    \begin{claim}\label{claim:inout}
    Let $S \subseteq V(H)$. Then $|\delta^+(S)| = |\delta^-(S)|$, where
    $\delta^+(S)$ and $\delta^-(S)$ are the sets of arcs from $S$ to $V(H)
    \setminus S$ resp.~from $V(H) \setminus S$ to $S$.
    \end{claim}
    In particular, Claim~\ref{claim:inout} implies that for each degree-$4$ vertex $v$ of $H$, $\indeg{v}=\outdeg{v}=2$. 
    Note that applying Rule~\ref{rule:join} to $v$ transforms $v$ into a degree-$2$ vertex 
    and that the in- and out-degrees of all other vertices stay unchanged. 
    Therefore we can apply Rule~\ref{rule:join} to every degree-$4$ vertex of $H$ and afterwards each vertex has a degree of exactly $2$, \ie,  $H'$ is a collection of cycles. 
   
    For the remaining algorithm, we need some insights regarding the existence
    of entry sites and exit points in $H'$. 
    To this end we first have to verify that after the while loop, $K$ is a cactus.
    Initially $K$ is a cycle and therefore a cactus. 
    Reversing Rule~\ref{rule:join} at an entry site does the following.
    The cycle $C$ of $K$ that contains the entry site $a$ obtains an additional vertex $v$.
    The cycle in $H'$ containing $v$ stays unchanged, but now it shares $v$ with $K$.
    Both increasing the length of $C$ and attaching a cycle to $K$ that only shares a single vertex with $K$ preserve the properties of cactus graphs.

    \begin{claim}\label{claim:entry}
    Before each run of the while loop, $K$ has an entry site.
    \end{claim}
    \begin{proof}
    Since $H$ is a connected graph, due to Claim~\ref{claim:inout} there are
    arcs from $V(H') \setminus V(K)$ to $V(K)$ in $H$. Let $\hat{a}=(v,w)$ be one of these
    arcs. Since we only removed arcs from $H$ by applying Rule~\ref{rule:join},
    the degree of $v$ in $H$ is $4$ and there is a second arc
    $\hat{a}'=(v,w')$ in $H$ that is not in $H'$.
    We conclude that there is an arc $\bi{w}{w'}$ in $K$.
    Since there are no arcs in $E(H')$ between $V(K)$ and $V(H) \setminus V(K)$ and $K$ is a cactus,    
    we conclude that $\bi{w}{w'}$ is an entry site.
    \end{proof}

    By Claim~\ref{claim:entry}, the while loop finds an entry site in each run and terminates with $V(K)=V(G)$. 
    We now show that in the for loop, the vertex $v$ is well defined.
    \begin{claim}\label{claim:exit}
    After the while loop, each cycle of $K$ except $\hat{C}'$ has exactly one exit point.
    \end{claim}
    \begin{proof}
    The cactus $K$ is formed by reversing applications of Rule~\ref{rule:join} in order to add cycles.
    For each cycle $C$ that is added to $K$ in the while loop, there are vertices $w,w' \in V(K)$ and $v\in V(C)$ such that reversing the application of Rule~\ref{rule:join} at \bi{w}{w'} introduces the arcs $(v,w)$ and $(v,w')$.
    This means that $v$ is a cut vertex of the transformed $K$ and thus an exit point of $C$.
    Each cycle $C'$ of $H'$ except $\hat{C}$ is added to $K$ exactly once during the while loop and there is exactly one cycle $C''$ of $K$ with $V(C') \subseteq V(C'')$.
    We conclude that $C''$ has exactly one exit point.
    \end{proof} 
 
    Therefore also the for loop terminates. 
    We now analyze the last step of \alg.
    \begin{claim}\label{claim:connected}
    $K''$ is a Hamiltonian tour.
    \end{claim}
    \begin{proof}
    Let $j$ be the number of degree-$4$ vertices in $K'$ after the for loop.
    We observe that the order of applications of Rule~\ref{rule:join} does not change the resulting graph.
    Therefore we may assign indices $1,2,\dotsc,j$ to the degree-$4$ vertices of $K'$ such that Rule~\ref{rule:join} is applied iteratively in order of the indices.
    Let $K_i$ be the obtained graph after $i$ iterations.
    Since applications of Rule~\ref{rule:join} do not increase degrees of vertices, each $K_i$ is degree-4-bounded.
    We now show by induction that the following invariant is true.
    For each $i \ge 0$, $K_i$ is a cactus and for each degree-4 vertex $v$ of $K_i$, there are arcs $(v,u),(v,u')$ with $u$ and $u'$ in different blocks.

    Initially, the invariant is true since $K_0 = K'$ is a cactus and we oriented each cycle of $K'$ in one direction.
    Now suppose that $1 \le i \le j$ and that the invariant is true for $K_{i-1}$.
    Let $v'$ be the vertex such that $K_i$ is obtained by applying Rule~\ref{rule:join} to $v'$ in $K_{i-1}$.
    By the invariant, there are two blocks $C,C'$ of $K_{i-1}$ and vertices $w,w'$ such that $w \in V(C)$, $w' \in V(C')$, and $(v',w),(v',w') \in E(K_{i-1})$.
    Then $C,C'$ are edge disjoint cycles and thus removing both $(v',w)$ and $(v',w')$ leaves a connected graph.
    The new arc $\bi{w}{w'}$ closes a cycle with vertex set $V(C) \cup V(C')$, which implies that $K_i$ is $2$-edge connected.
    We conclude that $K_i$ is a cactus, since any edge contained in two different cycles of $K_i$ would also be contained in two different cycles of $K_{i-1}$.

    Each degree-4 vertex $v \neq v'$ of $K_{i-1}$ is also a cut vertices in $K_i$ since removing $v$ would leave two components only one of which containing the arcs changed by applying Rule~\ref{rule:join}.
    In particular, the property that there are arcs $(v,u),(v,u')$ with $u$ and $u'$ in different blocks of $K_i$ did not change.
    Thus the invariant also holds for $K_i$.

    We conclude that $K''$ is a cactus with only vertices of degree 2 and therefore a Hamiltonian tour.
    \end{proof}

    To see that \alg runs in polynomial time, observe that
    each separate step can be done in polynomial time. The iterated application of
    Rule~\ref{rule:join} in the beginning and the end is done at most once for
    each vertex. Similarly, the while loop and the for loop are run at most
    linearly often.

    We finish the proof by analyzing the approximation ratio of \alg.
    We aim to partition the edges of $H$ into paths such that for each edges $e$ of $K''$ there is a path in $H$ between the ends of $e$.

    First we define a family $\mathcal{P}'$ of classes $P'_e \subseteq E(H)$, for each $e \in E(K')$.
    If $e \in E(H)$ we set $P'_e = \{e\}$.
    Otherwise, if $e \notin E(H)$, $e$ must have been obtained by applying Rule~\ref{rule:join}, replacing two edges $e',e''$ and we set $P'_e = \{e',e''\}$.

    We claim that $\mathcal{P}'$ is a partition of $E(H)$.
    Let us fix an edge $e' \in E(H)$.
    Then either $e' \in E(K')$ and thus $P'_{e'} = \{e'\}$, or there is an $e''$ such that
    $e',e''$ were replaced by $e$.
    But then $e \in E(K')$ since $e = \bi{u}{v}$ for two vertices $u,v$ and thus Rule~\ref{rule:join} cannot replace $e$.
    Therefore $P'_e = \{e',e''\}$. 
    Since $e''$ also cannot have two heads, $e'' \in E(H)$.
    Furthermore, $e'$ uniquely determines its class in $\mathcal{P}'$ and therefore it cannot be in two different classes. 

    Analogously, we define a family $\mathcal{P}''$ of classes $P''_e \subseteq E(K')$, for each $e \in E(K'')$.
    If $e \in K'$, we set $P''_e = \{e\}$.
    Otherwise, if $e \notin E(K')$, $e$ must have been obtained by applying Rule~\ref{rule:join}, replacing two edges $e',e''$ and we set $P''_e = \{e',e''\}$.
    With the same arguments as for $\mathcal{P}'$, $\mathcal{P}''$ is a partition of $E(K')$.

    We now define a partition $\mathcal{P}$ of $E(H)$ by combining $\mathcal{P}'$ and $\mathcal{P}''$.
    For each edge $e \in E(K'')$, we define the class $P_e$ of $\mathcal{P}$ as $\bigcup_{e' \in P''_e} P'_{e'}$.

    Let $E$ be the set of all edges that in some iteration of the for loop of \alg were denoted by $e$.
    Analogously, let $E'$ be the set of all edges that in some iteration of the for loop of \alg were denoted by $e'$.    
    \begin{claim}
    \label{claim:H-edge}
    For each $f \in E(K'') \setminus E(K')$,
    there are vertices $s,v,t$ such that 
    $P''_f = \{\{s,v\},\{v,t\}\}$ with $\{s,v\} \in E(H)$, $\{v,t\} \in E$, and $f = \{s,t\}$. 
    \end{claim}
    \begin{proof}
    Since $f \notin K'$, $f$ was obtained by applying Rule~\ref{rule:join} to a vertex $v$. 
    Therefore $P''_f = \{f',f''\} \subseteq E(K')$ where $f'$ and $f''$ have $v$ as common vertex.
    By the construction of $K''$, in $K'$ there are two blocks $C' \neq C''$ such that $f' \in E(C')$ and $f'' \in E(C'')$ 
    (see the proof of Claim~\ref{claim:connected}).
    For one of the two blocks, $v$ is its exit point and by renaming we assume that $C''$ is that block.
    Therefore in $K$, the two arcs $f',g' \in E(C')$ incident to $v$ have tails at $v$.
    In particular, $f'$ was not obtained from an application of Rule~\ref{rule:join} since the application only introduces arcs with two heads
    and thus $f' \in E(H)$.
    Due to the orientation in the for loop of $\alg$, we conclude that $f'' \in E$.
    The claim follows by setting $f' = \{s,v\}$ and $f'' = \{v,t\}$.
    \end{proof}

    We now show that $\mathcal{P}$ partitions $E(H)$ into short paths.
    \begin{claim}
    \label{claim:path}
    For each $e = \{s,t\} \in E(K'')$, the edges in $P_e$ form a path between $s$ and $t$ of length at most three in $H$.
    \end{claim}
    \begin{proof}
    Let us first assume that $e \in E(K')$.
    Then either $e \in E(H)$ or $P_e = \{e',e''\}$, where $e'=\{s,v\}$ and $e''=\{v,t\}$ for some vertex $v$.
    Since $e',e'' \in E(H)$, the claim follows.

    Otherwise, if $e \notin E(K')$, there are edges $f',f'' \in P''_e$ where $f' = \{s,v'\}$ and $f'' = \{v',t\}$ for some vertex $v'$.
    By Claim~\ref{claim:H-edge} we can rename the vertices and edges such that $f' \in E(H)$ and $f'' \in E$.
    Then either $f'' \in E(H)$ or there is a vertex $u''$ such that $P'_{f''} = \{\{v',u''\},\{u'',t\}\}$ with $\{v',u''\},\{u'',t\} \in E(H)$. 
    In both cases the claim follows.
    \end{proof}

    We now analyze the edge weights relative to $c(H)$. 
    For edges not in $H$ we have to consider a factor that is bounded due to the relaxed triangle inequality.  
    Let us fix an edge $f \in E(K'')$.
    We distinguish the possible sizes of $P_f$.
    If $|P_f| = 1$, $e \in E(H)$ and thus there is no factor to be considered.
    If $|P_f| = 2$, $P_f = \{f',f''\}$ for two adjacent edges $f',f'' \in E(H)$ and by the relaxed triangle inequality, 
    $c(f) \le \beta (c(f') + c(f''))$.
    
    If $|P_f| = 3$,
    by Claim~\ref{claim:H-edge} we have that $c(f) \le \beta \cdot (c(\hat{f}) + c(e))$ for an $\hat{f} \in E(H)$ 
    and an $e \in E$ with $P''_f = \{ \hat{f},e\}$.

    Let $e'\in E'$ be the edge that was compared to $e$ in the for loop.
    Since in the for loop each cycle is oriented only once, $E \cap E' = \emptyset$ and thus $e' \notin E$.

    We have that $c(e) \le \beta \cdot \min\{c(P'_e),c(P'_{e'})\}$, since otherwise $c(e') < c(e)$.
    Without loss of generality we may assume that $c(P'_e) \le c(P'_{e'})$ since otherwise the weight of $K''$ can only decrease.
    Then $c(e) \le \beta c(P'_e)$ and thus $c(f) \le \beta c(\hat{f}) + \beta^2 c(P'_e)$.
    
    Let $M = \bigcup_{g \in E} P'_g$ and $M' = \bigcup_{g \in E(K') \setminus E} P'_g$ and thus $c(M) \le c(M')$.
    By Claim~\ref{claim:path} we do not have to consider $|P_f| > 3$ and since we assumed that only edges in $M$ can have a factor $\beta^2$, 
    $c(K'') \le \beta \cdot c(M') + \beta^2 \cdot c(M)$.
    Since $\mathcal{P}$ is a partition of $E(H)$, by Theorem~\ref{thm:eul} 
    $c(K'') \le \beta c(H)/2 + \beta^2 c(H)/2 \le (3\beta/4 + 3\beta^2/4) \cdot \mathrm{opt}$,
    where $\mathrm{opt}$ is the weight of an optimal \TSP solution in $G$.
   
\section{Conclusion}\label{sec:conclusion}
We have seen that in a complete edge-weighted graph we can obtain a degree-$4$-bounded Eulerian
graph that has at most $1.5$ times the weight of an optimal TSP solution. This
result might allow for improvement. However, by setting $\beta=1$ an improvement would
directly imply an improvement over the approximation ratio of $1.5$ of the metric traveling
salesman problem and thus it would require valuable new insights.

The main result shows an approximability of the TSP with relaxed triangle
inequality that is better than all previous results as long as the value of $\beta$
is smaller than $13/3$. The most likely improvement of our result is for the term
$3\beta^2/4$, as it is not linear in $\beta$.
We conjecture the following.
\begin{conjecture}\label{con:alg}
    There is a polynomial time $(1.5 \beta)$-approximation algorithm for \TSP.
\end{conjecture}
Such a result seems likely as the
$4\beta$-approximation algorithm does not include a factor $\beta^2$ and the
factor $\beta^2$ in our result is only required in specific cases.
Thus both proving and disproving the conjecture would lead to interesting insights.

\section*{Acknowledgment}
I would like to thank Andr\'{a}s Seb\H{o} for valuable pointers to the
literature.

\newcommand{\etalchar}[1]{$^{#1}$}

\end{document}